\documentclass[a4paper]{article}
\pdfoutput=1

\usepackage{microtype}%if unwanted, comment out or use option "draft"
\usepackage{tikz}
\usetikzlibrary{arrows,automata,positioning,shapes}
\usepackage[utf8]{inputenc}
\usepackage{cite}
\usepackage{centernot}
\usepackage{color}
\usepackage{graphicx}
\usepackage{pgf}
\usepackage{url}
\usepackage[toc,page]{appendix}
\usepackage{booktabs}
\usepackage{xspace}
\usepackage{mathtools}
\usepackage{stmaryrd}
\usepackage{amsfonts,amsmath,amsthm}
\usepackage{thmtools}
\usepackage{xcolor}

\tikzstyle{vertex}=[auto=left,circle,draw=black]
\tikzstyle{label}=[auto=left,draw=none,node distance=0.5cm,font=\footnotesize]

\newcommand{\defparproblem}[4]{
  \vspace{1mm}
\noindent\fbox{
  \begin{minipage}{0.8\textwidth}
  \begin{tabular*}{0.9\textwidth}{@{\extracolsep{\fill}}lr} #1\end{tabular*}\\
  {\bf{Input:}} #2  \\
   {\bf{Parameter:}} #3 \\ 
  {\bf{Problem:}} #4
  \end{minipage}
  }
  \vspace{1mm}
}

\newtheorem{theorem}{Theorem}
\newtheorem{definition}{Definition}
\newtheorem{lemma}{Lemma}

\newtheorem*{IL}{Isolation Lemma}

\def\Iver#1{ \llbracket #1 \rrbracket }

\usepackage{schemata}

\usepackage[ruled,longend,vlined,linesnumbered]{algorithm2e}
\usepackage{wrapfig}
\usepackage{complexity}
\usepackage{paralist}

\title{A Fast Parameterized Algorithm for Co-Path Set}

\author{Blair D. Sullivan, Andrew van der Poel \\
    \footnotesize
    Department of Computer Science, North Carolina State University, \\
    \footnotesize Raleigh, NC 27695\\
    \footnotesize
    \texttt{\{blair\_sullivan,ajvande4\}@ncsu.edu}
}

\begin{document}

\def\mainAlg{\texttt{copath}\xspace}
\def\twAlg{\texttt{tw-copath}\xspace}

	\maketitle	
	\begin{abstract}
		The \textsc{$k$-Co-Path Set} problem asks, given a graph $G $ and a positive integer $k$, whether one can delete $k$ edges from $G$ so that the remainder is a collection of disjoint paths. We give a linear-time fpt algorithm with complexity $O^*(1.588^k)$ for deciding \textsc{$k$-Co-Path Set}, significantly improving the previously best known $O^*(2.17^k)$ of Feng, Zhou, and Wang (2015). Our main tool is a new $O^*(4^{tw(G)})$ algorithm for \textsc{Co-Path Set} using the Cut\&Count framework, where $tw(G)$ denotes treewidth. In general graphs, we combine this with a branching algorithm which refines a $6k$-kernel into reduced instances, which we prove have bounded treewidth.

	\end{abstract}	
	\section{Introduction}\label{sec:intro}
	We study parameterized versions of \textsc{Co-Path Set}~\cite{cpsreduction, zhang}, an NP-complete problem asking 
for the minimum number of edges whose deletion from a graph results in a collection of disjoint paths (the deleted edges being a {\em co-path set} -- see Figure~\ref{fig:min_ex}).  
Specifically, we are concerned with \textsc{$k$-Co-Path Set}, which uses the natural parameter of the number of edges deleted. \looseness-1

\begin{center}
\defparproblem{\textsc{$k$-Co-Path Set}}{A graph $G = (V,E)$ and a non-negative integer $k$.}{$k$}{Does there exist $F \subseteq E$ of size exactly $k$ such that $G[E\setminus F]$ is a set of disjoint paths?}
\end{center}

These problems are naturally motivated by determining the ordering of {genetic markers} in DNA using fragment data created by breaking chromosomes with gamma radiation (a technique known as \textit{radiation hybrid mapping})~\cite{cox, richard, slonim}. Unfortunately, human error in distinguishing markers often means the constraints implied by markers' co-occurrence on fragments are incompatible with all possible linear orderings, necessitating an algorithm to find the ``best'' ordering (that violates the fewest constraints).
\textsc{Co-Path Set} solves the special case where each DNA fragment contains exactly two genetic markers (corresponding to an edge in the graph); any linear ordering of the markers must correspond to some set of paths, and we minimize the number of unsatisfied constraints (edges in the co-path set). 

 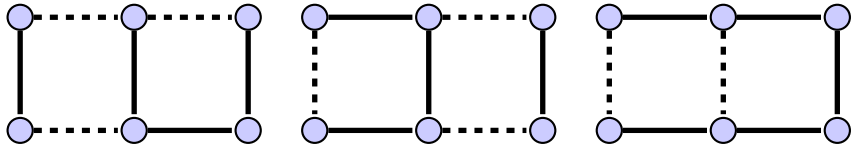
\begin{figure}
	\centering
	\begin{minipage}{0.31\textwidth}
\begin{tikzpicture}[>=stealth',shorten >=1pt,auto,node distance=1.5cm,
  thick,main node/.style={circle,fill=blue!20,draw,font=\sffamily\bfseries},path/.style={rectangle,fill=blue!20,draw,font=\sffamily\bfseries},bunch/.style={ellipse,fill=blue!20,draw,font=\sffamily\bfseries}]

\node[main node] (1) {};
  \node[main node] (2) [below of=1] {};
  \node[main node] (3) [right of=1] {};
      \node[main node] (4) [below of=3]{};
	\node[main node] (5) [right of=3] {};
      \node[main node] (6) [below of=5]{};

  \draw[every node/.style={font=\sffamily\small}, line width = 0.7mm]
    (1) edge  node[above] {} (2)
        edge [dashed] node [below] {} (3)
    (2) edge [dashed] node [above] {} (4)
    (3) edge  node [above] {} (4)
        edge [dashed] node [below] {} (5)
      (4)  edge node [below] {} (6)
       (5) edge node [below] {} (6);

\end{tikzpicture}
\end{minipage}
\begin{minipage}{0.31\textwidth}
\begin{tikzpicture}[>=stealth',shorten >=1pt,auto,node distance=1.5cm,
  thick,main node/.style={circle,fill=blue!20,draw,font=\sffamily\bfseries},path/.style={rectangle,fill=blue!20,draw,font=\sffamily\bfseries},bunch/.style={ellipse,fill=blue!20,draw,font=\sffamily\bfseries}]

\node[main node] (1) {};
  \node[main node] (2) [below of=1] {};
  \node[main node] (3) [right of=1] {};
      \node[main node] (4) [below of=3]{};
	\node[main node] (5) [right of=3] {};
      \node[main node] (6) [below of=5]{};

  \draw[every node/.style={font=\sffamily\small}, line width = 0.7mm]
    (1) edge[dashed]  node[above] {} (2)
        edge node [below] {} (3)
    (2) edge node [above] {} (4)
    (3) edge  node [above] {} (4)
        edge [dashed] node [below] {} (5)
      (4)  edge [dashed] node [below] {} (6)
       (5) edge node [below] {} (6);

\end{tikzpicture}

\end{minipage}
\begin{minipage}{0.31\textwidth}
\begin{tikzpicture}[>=stealth',shorten >=1pt,auto,node distance=1.5cm,
  thick,main node/.style={circle,fill=blue!20,draw,font=\sffamily\bfseries},path/.style={rectangle,fill=blue!20,draw,font=\sffamily\bfseries},bunch/.style={ellipse,fill=blue!20,draw,font=\sffamily\bfseries}]

\node[main node] (1) {};
  \node[main node] (2) [below of=1] {};
  \node[main node] (3) [right of=1] {};
      \node[main node] (4) [below of=3]{};
	\node[main node] (5) [right of=3] {};
      \node[main node] (6) [below of=5]{};

  \draw[every node/.style={font=\sffamily\small}, line width = 0.7mm]
    (1) edge[dashed]  node[above] {} (2)
        edge node [below] {} (3)
    (2) edge node [above] {} (4)
    (3) edge [dashed]  node [above] {} (4)
        edge node [below] {} (5)
      (4)  edge node [below] {} (6)
       (5) edge node [below] {} (6);

\end{tikzpicture}

\end{minipage}
	\caption{Three co-path sets (dashed edges), including one of minimum size (rightmost).}
	\label{fig:min_ex}
\end{figure}

Recent algorithmic results related to \textsc{Co-Path Set} include a $(10/7)$-approximation algorithm~\cite{chen}, and two parameterized algorithms deciding \textsc{$k$-Co-Path Set}~\cite{feng, feng2}, the faster of which~\cite{feng2} has time complexity\footnote{Throughout this paper, we use the notation $O^*(f(k))$ for the fpt (fixed-parameter tractable) complexity $O(f(k) n^{O(1)})$; we say an algorithm is {\em linear-fpt} if the complexity is $O(f(k) n)$.} $O^*(2.17^k)$.  However, as written, both parameterized results~\cite{feng, feng2} contain a flaw in their analysis which invalidates their probability of a correct solution in the given time\footnote{Step 2.11 in both versions of Algorithm R-MCP checks if a candidate co-path set $F$ has size $\leq k_1$ (as they are sweeping over all possible sizes of candidates and want to restrict the size accordingly). If $F$ is too large, the algorithm discards it and continues to the next iteration. However, in order for their analysis to hold, the probability that the candidate is contained in a co-path set must be $\geq (1/2.17)^{k_1}$ (or $(1/2.29)^{k_1}$ in~\cite{feng}) for \emph{every} iteration. Candidates which are too large may have significantly smaller probability of containment, yet are counted in the exponent of the analysis.}.
The best known bound prior to~\cite{feng} is an $O^*(2.45^k)$ algorithm~\cite{zhang}. 
In this paper, we prove:  \looseness-1
 \begin{theorem}
 \textsc{$k$-Co-Path Set} is decidable in $O^*(1.588^k)$ linear-fpt time with probability at least $2/3$.
 \label{main}
 \end{theorem} 
We note that standard amplification arguments apply, and Theorem \ref{main} holds for any success probability less than $1$. Further, if $f$ is an increasing function with $\lim_{n \rightarrow \infty} f(n) = 1$, we can solve \textsc{$k$-Co-Path Set} with success probability at least $f(n)$ in $O(1.588^kn\polylog(n))$. 
 
The remainder of this paper is organized as follows: after essential definitions and notation in Section~\ref{sec:prelim}, we start in Section~\ref{sec:cutcount} by giving a new $O^*(4^{tw(G)})$ algorithm \twAlg for solving \textsc{Co-Path Set} parameterized by treewidth ($tw$) using the Cut\&Count framework~\cite{cygan}. Finally, Section \ref{sec:result} describes the linear-fpt algorithm referenced in Theorem \ref{main}, which solves \textsc{$k$-Co-Path Set} on general graphs  in $O^*(1.588^k)$ by applying \twAlg to a set of ``reduced instances'' generated via kernelization and a branching procedure\footnote{The properties of our reduced instances guarantee we can find a tree decomposition in poly($k$) time.} \texttt{deg-branch}.

	\section{Preliminaries}\label{sec:prelim}
	Let $G(V,E)$ be the graph with vertex set $V$ and edge set $E$. Unless otherwise noted, we assume $|V| = n$;  we let $N(v)$ denote the set of neighbors of a vertex $v$, and let $\deg(v) = |N(v)|$. Given a graph $G(V,E)$ and $F \subseteq E$, we write $G[F]$ for the graph $G(V,F)$. 

Our \twAlg algorithm in Section~\ref{sec:cutcount} uses dynamic programming over a \textit{tree decomposition}, and its running time depends on the related measure of \textit{treewidth}~\cite{robertson}, which we denote $tw(G)$. To simplify the dynamic programming, we will use a variant of \emph{nice tree decompositions}~\cite{ntd, cygan} where each node in the tree has one of five specific types: leaf, introduce vertex, introduce edge, forget vertex, or join.  The ``introduce edge'' nodes are labelled with an edge $uv$ and have one child (with an identical bag); we require each edge in $E$ is introduced exactly once. Additionally, we enforce that the root node is of type ``forget vertex'' (and thus has an empty bag). A tree decomposition can be transformed into a nice decomposition of the same width in time linear in the size of the input graph~\cite{cygan}.  

When describing the dynamic programming portion of the algorithm we use Iverson's bracket notation:  if $p$ is a predicate we let $\Iver{p}$ be 1 if $p$ is true and 0 otherwise. We also use the shorthand $f[x\to y]$ to denote updating a function $f$ so that $f(x) = y$ and all other values are unchanged.  

Finally, we use fast subset convolution~\cite{convolution1} to reduce the complexity of handling join nodes in the nice tree decomposition (Section~\ref{sec:cutcount}). This technique maps functions of the vertices in a join bag to values in $\mathbb{Z}_p = \mathbb{Z}/p\mathbb{Z}$ (where $p$ is chosen based on the application). The key complexity result we rely on uses the $\mathbb{Z}_p$ product, which is defined below.  We write $\mathbb{Z}_p^B$ for the set of all vectors $t$ of length $|B|$ assigning a value $t(b) \in \mathbb{Z}_p$ to each element of $b \in B$. 

\begin{definition}[$\mathbb{Z}_p$ product]
Let $p \geq 2$ be a fixed integer and let B be a finite set. For $t_1, t_2, t \in \mathbb{Z}_p^B$ we say that $t_1 + t_2 = t$ if $t_1(b) + t_2(b) = t(b)$ (in $\mathbb{Z}_p$) for all $b \in B$. For a ring $R$ and functions $f, g: \mathbb{Z}_p^B \rightarrow R$, define the $\mathbb{Z}_p$ product, $*_x^p$ as
\begin{center}
$(f *_x^p g)(t) = \sum\limits_{t_1 + t_2 = t} f(t_1)g(t_2)$.
\end{center}
\end{definition}

\noindent Fast subset convolution guarantees that certain $\mathbb{Z}_p$ products can be computed quickly. 

\begin{lemma}[Cygan et al.~\cite{cygan}]
Let $R = \mathbb{Z}$ or $R = \mathbb{Z}_q$ for some constant q. The $\mathbb{Z}_4$ product of functions $f, g: \mathbb{Z}_4^B \rightarrow R$ can be computed in time and ring operations $4^{|B|}|B|^{O(1)}$.
\label{lem:Zp_prod}
\end{lemma}

	\section{An $O^*(4^{tw(G)})$ Algorithm via Cut\&Count}\label{sec:cutcount}
	We start by giving an fpt algorithm for \textsc{Co-Path Set} parameterized by treewidth. Our primary tool is the the Cut\&Count framework, which enables $c^{tw}n^{O(1)}$ one-sided Monte Carlo algorithms for connectivity-type problems with constant probability of a false negative. Cut\&Count has previously been used to improve the best-known bounds for several well-studied problems, including \textsc{Connected Vertex Cover}, \textsc{Hamiltonian Cycle}, and \textsc{Feedback Vertex Set}~\cite{cygan}. Pilipczuk showed that an $O^*(c^{tw})$ algorithm for some constant $c$ can be designed with the Cut\&Count approach for \textsc{Co-Path Set} because the problem can be expressed in the specialized graph logic known as ECML+C~\cite{pilipczuk}. However, since our end goal is to improve on existing algorithms for \textsc{$k$-Co-Path Set} in general graphs using a bounded treewidth kernel, we need to develop a specialized dynamic programming algorithm with a small value of $c$. We show: \looseness-1

\begin{theorem}\label{cutcount}
There exists a one-sided fpt Monte Carlo algorithm \emph{\texttt{tw-copath}} deciding \textsc{$k$-Co-Path Set} for all $k$ in a graph $G$ in $O^*(4^{tw(G)})$ time with failure probability $\leq1/3$. 
\end{theorem}

The Cut\&Count technique has two main ingredients: an algebraic approach to counting which uses arithmetic in $\mathbb{Z}_2$ (enabling faster algorithms) alongside a guarantee that undesirable objects are seen an even number of times (so a non-zero result implies a desired solution has been seen), and the idea of defining the problem's connectivity requirement through consistent cuts.  In this context, a {\em consistent cut} is a partitioning $(V_1, V_2)$ of the vertices of a graph into two sets such that no edge $uv$ has $u \in V_1$ and $v \in V_2$ {and} all vertices of degree $0$ are in $V_1$ . Since each connected component must lie completely on one side of any consistent cut, we see that a graph $G$ has exactly $2^{cc(G) - n_I(G)}$ such cuts, where $cc(G)$ is the number of connected components and $n_I(G)$ is the number of isolated vertices. 
In order to utilize parity with the number of consistent cuts, we introduce \textit{markers}, which create even numbers of consistent cuts for graphs that are not collections of disjoint paths. 
Our counting algorithm \twAlg, which computes the parity of the size of the collection of subgraphs with consistent cuts which adhere to specific properties pertaining to \textsc{Co-Path Set}, employs dynamic programming over a nice tree decomposition. 
We further use weights and the Isolation Lemma to bound the probability of a false negative arising from multiple valid markings of a solution. We use fast subset convolution~\cite{convolution1} to reduce the complexity required for handling join bags in the dynamic programming.  In the remainder of this section, we present the specifics for applying these techniques to solve \textsc{Co-Path Set}. \looseness-1
 
\subsection{Cutting}

We first provide formal definitions of markers and marked consistent cuts, which we use to ensure that sets of disjoint paths are counted exactly once in some entry of our dynamic programming table. 

 \begin{definition} A triple $(V_1, V_2, M)$ is a \textbf{marked consistent cut} of a graph $G$ if $(V_1, V_2)$ is a consistent cut and  $M \subseteq E(G[V_1])$. We refer to the edges in $M$ as the \textbf{markers}.  A marker set is \textbf{proper} if it contains at least one edge in each non-isolate connected component of $G$.
 \label{def:mcc}
\end{definition}

Note that if a marker set is proper, all vertices are on the $V_1$ side of the cut. This is because by the definition of a consistent cut, all isolates are on the $V_1$ side, and if every connected component contains a marker then all connected components must fall entirely on the $V_1$ side as well. Therefore for any proper marker set there exists exactly one consistent cut, while all marker sets which are not proper will be paired with an even number of consistent cuts because unmarked components may lie in $V_1$ \emph{or} $V_2$.
We use proper marker sets to distinguish desired subgraphs by assigning markers in such a way that when we prune the dynamic programming table for solutions (as described later in the section), the only subgraphs we consider which may have a proper marker set are collections of disjoint paths. We know because the marker set is proper that the subgraph has a unique consistent cut, and thus these collections of disjoint paths will only be counted once in some entry of the dynamic programming table, while all other subgraphs will be counted an even number of times. Note that we are not claiming that all collections of disjoint paths will have proper marker sets.
%In the definition of our dynamic programming we force all isolates to be on the $V_1$ side of the cut.

We refer to the complement of a co-path set (the edges in the disjoint paths) as a \textit{cc-solution}, and call it a \textit{marked-cc-solution} when paired with a proper marker set of size exactly equal to its number of non-isolate connected components. While cc-solutions can be viewed as solutions due to their complementary nature, being marked is crucial in our counting algorithm and thus subgraphs which are marked-cc-solutions are what correspond to solutions in the dynamic programming table. \looseness-1

We now describe our use of the Isolation Lemma, which guarantees we are able to use parity to distinguish solutions. Let $f(X)$ denote $\sum_{x \in X} f(x)$. 

\begin{IL}[\hspace*{-3pt}\cite{isolation}]  Let $\mathcal{F} \subseteq 2^U$ be a non-empty set family over universe $U$.  A function $\omega\colon U \to \mathbb{Z}$ is said to \emph{isolate} $\mathcal{F}$ if there is a unique $S \in \mathcal{F}$ with $\omega(S) =\min_{F \in \mathcal{F}} \omega(F)$.  Assign weights $\omega \colon U \to \{1, 2, ... , N\}$ uniformly at random. Then the probability that $\omega$ isolates $\mathcal{F}$ is at least $1 - |U| / N$.
\label{iso_lemma}
\end{IL}

Intuitively, if  $\mathcal{F}$ is the set of solutions (or complements of solutions) to an instance of \textsc{Co-Path Set} and $|\mathcal{F}|$ is even, then \twAlg would return a false negative. This is because while each solution is counted an odd number of times in \twAlg, because there are an even number of solutions the total count of solutions is even, making the combined count of solutions and non-solutions even and the algorithm would incorrectly determine a solution does not exist (a false negative). The Isolation Lemma allows us to partition $\mathcal{F}$ based on the weight of each solution (as assigned by $\omega$), and guarantees at least one of the partition's blocks has odd size with constant probability. 
We let $U$ contain two copies of every edge $e \in E$: one representing $e$ as a marker and one as an edge in the cc-solution. Then $2^U$ denotes all pairs of edge subsets (potential marked-cc-solutions), and we set $N = 3 |U| = 6 E$ (selected to achieve success probability in Theorem \ref{main}).
Each copy of an edge is assigned a weight in $[1, N]$ uniformly at random by $\omega$ and the probability of finding an isolating $\omega$ is thus $2/3$. We denote the values assigned by $\omega$ to the set of marker copies by $\omega_M$, and likewise to the set of edge-in-cc-solution copies by $\omega_E$.
 \looseness-1

\subsection{Counting}

A marked-cc-solution $C$ of a graph $G$ corresponds to a co-path set of size $k$ when the number of edges and markers in $C$ match specific values which depend on $k$ and $|E(G)|$. These values are easily deduced because we know the deletion of a co-path set solution of size $k$ will leave $|E(G)| - k$ edges in a cc-solution. Furthermore, because a forest has $n - m$ connected components, the number of markers in $C$ needs to be at most $|V(C)| - |E(G)| + k$. All isolates from a forest can be removed and the resulting graph is still a forest, and thus the actual number of markers necessary in $C$ is $|V(C)| - n_I(C) - |E(G)| + k$.

We now describe a dynamic programming (DP) algorithm over a nice tree decomposition which returns mod 2 the number of appropriately sized marked-cc-solutions in the root's subtree (for a fixed $k$).  Since no-instances have no appropriately sized marked-cc-solutions, and yes-instances have at least one, odd parity for the number of marked-cc-solutions of size corresponding to $k$ implies a solution to the \textsc{$k$-Co-Path Set} instance must exist.  

During the DP algorithm we actually count (for all values $(m, e)$) the number of \emph{cc-candidates}, which are subgraphs $G' \subseteq G$ with maximum degree 2, exactly $e$ edges, and a marked consistent cut with $m$ markers. The following lemma justifies counting cc-candidates in place of marked-cc-solutions.
\begin{lemma}
The parity of the number of marked-cc-solutions in $G$ with $e$ edges and weight $w$ is the same as the parity of the number of cc-candidates $G' \subseteq G$ with $e$ edges, $|V(G')| - e - n_I(G')$ markers, and weight $w$. 
\label{lem:equality}
\end{lemma}
\begin{proof}
Consider a subgraph $G' \subseteq G$ with maximum degree 2 and $e$ edges. Let $M'$ be a marking of $G'$ such that $\omega_E(E(G')) + \omega_M(M') = w$. Assume first that $G'$ is a collection of paths. We know that $G'$ has $|V(G')| - e - n_I(G')$ non-isolate connected components. If $M'$ is a proper marker set of $G'$, then $|M'| = |V(G')| - e - n_I(G')$ and $(G', M')$ has exactly one consistent cut. Therefore $(G', M')$ contributes one to both the number of marked-cc-solutions and the number of cc-candidates, respectively.  
   
If otherwise $M'$ is not a proper marker set, then $(G', M')$ contains an unmarked connected component and has an even number of consistent cuts, and therefore contributes an even number to the count of cc-candidates and zero to the number of marked-cc-solutions. Finally, if $G'$ contains at least one cycle then $cc(G') >|V(G')| - e - n_I(G')$. Therefore at least one connected component does not contain a marker, and the number of consistent cuts is even, so the contribution to the count of cc-candidates is again even and the contribution to the count of marked-cc-solutions is zero. We conclude that the parity of the number of marked-cc-solutions and the parity of the number of cc-candidates is the same.
\end{proof}

Our dynamic programming algorithm is a bottom-up approach over a nice tree decomposition. We build cc-candidates for all values of $m$ and $e$ (encoding the option to add/not add edges and select/not select edges as markers), and keep track of various parameters ensuring that when pruning the DP table we only consider cc-candidates which could be valid solutions to the \textsc{$k$-Co-Path Set} instance.
We use the number of edges to ensure our solution is of the correct size, and the number of markers and non-isolate vertices to determine when a subgraph is acyclic. The weight parameter allows us to distinguish between solutions and decreases the likelihood of a false negative occurring via the Isolation Lemma. Finally, we need a parameter that encodes the degree information required to properly combine cc-candidates as we iterate up the tree.
\smallskip
\begin{table}[thb]
\centering
\begin{tabular}{ p{1em} l l   p{2em} r p{1em}}
\toprule
  \multicolumn{2}{c}{Variable\hspace*{1.5em}} & Parameter &  \multicolumn{3}{c}{Maximum value} \\ \midrule
  & $a$ &  $\#$ of non-isolated vertices && $ n$ & \\
  & $e$ & $\#$ of edges && $n^2$ &\\
  & $m$ & $\#$ of markers && $n^2$ &\\
  & $w$ & weight of edges and markers && $4n^4$ &\\
  \bottomrule
\end{tabular}
\vspace{0.5em}
\caption{Dynamic programming table parameters and upper bounds. }
 \label{tab:values}
\end{table}
We call this parameter a \emph{degree-function} and define on the vertices $V$ of a bag as $f: V \rightarrow \Sigma = \{0, 1_1, 1_2, 2\}$, where $f(v)$ corresponds to $v$'s degree in the associated cc-candidates of the table entry  --- for vertices of degree $1$, their value $1_j$ denotes which side of the partition $(V_1,V_2)$ they are on. Vertices with degree $0$ are on the $V_1$ side of the cut by definition so we need not keep track of their side of the cut. Similarly degree $2$ vertices cannot have additional incident edges, thus the side of the cut they fall on will not matter for selecting markers. %and is additionally encoded in the endpoints of the path they are on. 
In summary, we have table entries $A_x(a, e, m, w, s)$ counting the number of cc-candidates at bag $x$ with $a$ non-isolated vertices, $e$ edges, $m$ markers, weight $w$, and degree-function $s$.

In the following description of the dynamic programming algorithm over a nice tree decomposition $T$, we let  $z_1, z_2$ denote the children of a join node; otherwise, the unique child is denoted $y$.  \\

\medskip
\noindent\textbf{Leaf}:
$$A_x(0, 0, 0, 0, \emptyset) = 1\text{; }A_x(a, e, m, w, s) = 0 \text{ for all other inputs.}$$

\noindent\textbf{Introduce vertex $v$}: 
$$A_x(a, e, m, w, s) = \Iver{s(v)=0} A_y(a, e, m, w, s).$$ 

\noindent\textbf{Introduce edge $uv$}: 
\begin{gather*}
A_x(a, e, m, w, s) = A_y(a, e, m, w, s) + 
\sum_{\mathclap{\substack{\alpha_t \in subs(s(t)) \\ t \in \{u,v\}}}}  \; \Iver{\phi_2(\alpha_u,\alpha_v)}A_y(a', e - 1, m, w', s') \\
+\; \sum_{\mathclap{\substack{\alpha_t \in subs(s(t)) \\ t \in \{u,v\}}}} \; \Iver{\phi_1(\alpha_u,\alpha_v)}\Big(A_y(a', e - 1, m, w', s') + A_y(a', e - 1, m - 1, w'', s')\Big), 
\end{gather*}

\noindent where $\phi_j(\alpha_u,\alpha_v) =  (\alpha_u = 1_j \lor s(u) = 1_j) \land (\alpha_v = 1_j \lor s(v) = 1_j)$, $a' = a - (| \{1_1, 1_2\} \cap \{s(u), s(v)\}|)$, $w' = w - \omega_E(uv)$,  $w'' = w - \omega_E(uv) - \omega_M(uv)$, $s' = s[u \to \alpha_u, v \to \alpha_v]$, and the $subs$ function returns all the values the degree-function in child node $y$ could have assigned to vertices $u$ and $v$ based on current degree-function $s$ (summarized below).

\begin{center}
\setlength{\tabcolsep}{0.5em}%
\noindent \begin{tabular}{ l  c  c  c  c }%
\toprule%
  $s(v)$ & 0 & $1_1$ & $1_2$ & 2\\ \midrule
  $subs(s(v))$ & $\emptyset$ & 0 & 0 & $\{1_1, 1_2\}$ \\ 
  \bottomrule
\end{tabular}
\end{center}

We now argue this formula's correctness. The term $A_y(a, e, m, w, s)$ handles the case when $uv$ is excluded from the cc-solution. We handle the case when $uv$ is added to the cc-solution by iterating over all possible \emph{subs} values for each endpoint, only considering counts in child $y$'s entries
where $u$ and $v$ have the appropriate \emph{subs} values (preventing us from ever having a vertex with degree greater than 2). 
Note that we use the $\phi_j$ function to guarantee that if $s$ labels $u$ or $v$ as an isolate, we do not use the introduced edge. We have a summation for both possible $j$ values in order to consider $uv$ falling on either side of the cut. 
The formulation of $a'$ assures that each endpoint of degree $1$ is now included in the count of non-isolates (i.e. when $u$ and/or $v$ had degree $0$ in $y$).
We utilize the marker weight of $uv$ to distinguish when we choose it as a marker (only if on $V_1$ side of cut), and increment $m$ accordingly.
In either case, we update $w$ appropriately (with $w'$ if no marker, $w''$ if marker introduced).  \looseness-1

\noindent\textbf{Forget vertex $h$:} \\
$$A_x(a, e, m, w, s) = \sum_{\alpha \in \{0, 1_1, 1_2, 2\}} A_y(a, e, m, w, s[h \to \alpha]).$$

As a forgotten vertex can have degree 0, 1 or 2 in a cc-candidate, we must consider all possible values that $s$ assigns to $h$ in child bag $y$. Note that cc-candidates in which $h$ is both not an isolate and not a member of a connected component that contains a marker will cancel mod $2$, as $h$ can be on either side of the cut and all parameters will be identical.

\noindent\textbf{Join:} \\
We compute $A_x$ from $A_{z_1}$ and $A_{z_2}$ via fast subset convolution~\cite{convolution1} taking care to only combine table entries whose degree-functions are \emph{compatible}. 
\begin{definition}
At a join node $x$ with children $z_1$ and $z_2$, the degree-functions $s_1$ from $A_{z_1}, s_2$ from $A_{z_2}$, and $s$ from $A_x$ are \textbf{compatible} if one of the following holds for every vertex $v$ in $x$: (i) $s_i(v) = 0 \text{ and } s_l(v) = s(v), i \neq l$  or (ii) $s_1(v) = s_2(v) = 1_j \text{ and } s(v) = 2$ for $i, j, l \in [1, 2]$. 
\end{definition}

In order to apply Lemma~\ref{lem:Zp_prod}, we let $B$ be the bag at $x$, and transform the values assigned by the degree function $s$ to values in $\mathbb{Z}_4$. Let $\phi \colon \{0, 1_1, 1_2, 2\} \to \mathbb{Z}_4$ and $\rho \colon \{0, 1_1, 1_2, 2\} \to \mathbb{Z}$ be defined as in the table below, extending to vectors by component-wise application.

\begin{center}
\setlength{\tabcolsep}{0.5em}%
\noindent \begin{tabular}{ p{10pt}  p{10pt} p{10pt}  p{10pt}  p{10pt} } 
   & 0 & $1_1$ & $1_2$ & 2 \\ \midrule
   $\phi$ & 0 & 1 & 3 & 2 \\ 
   $\rho$ & 0 & 1 & 1 & 2 
\end{tabular}
\end{center}

\noindent We use $\phi$ to apply Lemma~\ref{lem:Zp_prod}, while the function $\rho$ (which corresponds to a vertex's degree) is used in tandem to ensure the compatibility requirements are met: if $\phi(s_1) + \phi(s_2) = \phi(s)$, then necessarily $\rho(s_1) + \rho(s_2) \geq \rho(s)$. From the above table it is easy to verify that $\phi(s_1) + \phi(s_2) = \phi(s)$ and $\rho(s_1) + \rho(s_2) = \rho(s)$ together imply that $s_1, s_2 \text{ and } s$ are compatible. We sum over both functions when computing values for join nodes, to make sure that solutions from the children are combined only when there is compatibility. 

 Assign $t_1 = \phi(s_1)$, $t_2 = \phi(s_2)$, and $t = \phi(s)$ in accordance with Lemma \ref{lem:Zp_prod}.   Let $\rho(s) = \sum_{v \in B} \rho(s(v))$; that is $\rho(s)$ is the sum of the degrees of all the vertices in the join node, as assigned by degree-function $s$. By defining functions $f$ and $g$ as follows: 
\begin{align*}
f^{\langle d, a, e, m, w\rangle}(\phi(s)) &= \Iver{\rho(s) = d}A_{z_1}(a, e, m, w, s), \\ 
g^{\langle d,a, e, m, w\rangle}(\phi(s)) &= \Iver{\rho(s) = d}A_{z_2}(a, e, m, w, s),
\end{align*}

\noindent and writing $\vec{r_i}$ for the vector $\langle d_i,a_i,e_i,m_i,w_i\rangle$, we can now compute
\begin{equation*}
 A_x(a, e, m, w, s) = \sum_{\vec{r_1}+\vec{r_2} = \langle \rho(s),a',e,m,w\rangle} (f^{\vec{r_1}} *_x^4 g^{\vec{r_2}})(\phi(s)) 
\end{equation*}
where $a' = a + |s_1^{-1}\{1_1, 1_2\} \cap s_2^{-1}\{1_1, 1_2\}|$. We point out that 
\begin{equation*}
\sum_{\vec{r_1}+\vec{r_2} = \langle \rho(s),a',e,m,w\rangle} (f^{\vec{r_1}} *_x^4 g^{\vec{r_2}})(\phi(s))  = 1 
\end{equation*}
only if both $\phi(s_1) + \phi(s_2) = \phi(s)$ and $\rho(s_1) + \rho(s_2) = \rho(s)$; that is, exactly when $s_1, s_2 \text{ and } s$ are compatible.
  \\

\medskip 

We conclude this section by describing how we search the DP table for marked-cc-solutions at the root node $r$.
By Lemma~\ref{lem:equality}, the parity of the number of marked-cc-solutions with $|E| - k$ edges and weight $w$ is the same as the parity of the number of cc-candidates $G'$ with $|E| - k$ edges, $|V(G')| - (|E| - k) - n_I(G')$ markers and weight $w$. These candidates are recorded in the table entries $A_r(a, |E|-k, a-|E|+k, w, \emptyset)$, where $a$ is the number of non-isolates. Therefore, if there exists some $a$ and $w$ so that $A_r(a, |E|-k, a-|E|+k, w, \emptyset) = 1$, then we have a yes-instance of \textsc{$k$-Co-Path Set}. Note that the degree-function is $\emptyset$ in this entry because there are no vertices contained in the root node by definition.
 
By Lemma \ref{lem:Zp_prod}, the time complexity of \twAlg for a join node $B$ is  $O^*(4^{|B|})$, which is $O^*(4^{tw})$. 
Note that for the other four types of bags, as we only consider one instance of $s$ per table entry, the complexity for each is $O^*(4^{tw})$. We point out that the size of the table is polynomial in $n$ because there are a linear number of bags and a polynomial number of entries (combinations of parameters) for each bag.
Since the nice tree decomposition has size linear in $n$, the bottom-up dynamic programming runs in total time $O^*(4^{tw})$. This complexity bound combined with the correctness of \twAlg discussed above proves Theorem \ref{cutcount}.

	\section{Achieving $O^*(1.588^k)$ in General Graphs}\label{sec:result}
	In order to use \twAlg to solve \textsc{k-CoPath Set} in graphs with unbounded treewidth, we combine kernelization and a branching procedure to generate a set of \emph{reduced instances} -- bounded treewidth subgraphs of the input graph $G$.  Specifically, we begin by constructing a kernel of size at most $6k$ as described in~\cite{feng2}. Our reduced instances are bounded degree subgraphs of the kernel given by a branching technique. We prove that (1) at least one reduced instance is an equivalent instance; (2) we can bound the number of reduced instances; and (3) each reduced instance has bounded treewidth. Finally, we analyze the overall computational complexity of this process.

\subsection{Kernelization and Branching}

We start by describing our branching procedure \texttt{deg-branch} (Algorithm \ref{alg:deg-branch}), which uses a degree-bounding technique similar to that of Zhang et al.~\cite{zhang}. Our implementation takes an instance $(G,k)$ of \textsc{Co-Path Set} and two non-negative integers $\ell$ and $D$, and returns a set of reduced instances $\{(G_i, k-\ell)\}$ so that (1) each $G_i$ is a subgraph of $G$ with exactly $|E|-\ell$ edges and maximum degree at most $D$; and (2) at least one $(G_i, k-\ell)$ is an equivalent instance to $(G,k)$. The size of the output (and hence the running time) of \texttt{deg-branch} depends on both input parameters $\ell$ and $D$. We will select $D$ to achieve 
the desired complexity in $\mainAlg$ in Section~\ref{sec:copathalg}. We also make use of a budget parameter $b$, which keeps track of how many more edges can be removed per the constraints of $\ell$ ($b$ is initially set to $\ell$). 

Our branching procedure leverages the observation that if a co-path set $S$ exists, then every vertex has at most two incident edges not in $S$.  Specifically, for every vertex of degree greater than $D$, we branch on pairs of incident edges which could remain after removing a valid co-path set (calling each pair a \emph{candidate}), creating a search tree of subgraphs.  

 \def\kRec{\ell}
\begin{algorithm}[htb]
\SetKwFunction{bnddeg}{deg-branch}
	\DontPrintSemicolon
	\SetKwProg{bnddegprog}{Algorithm}{}{}
	\bnddegprog{  \bnddeg{$G,k,\kRec,D,b$}  }{
		Let $v$ be a vertex of maximum degree in $G$\;
		\uIf{$deg(v) \geq D+1$ \emph{and}~$b \geq D-1$}{
			Select vertices $u_1, \ldots u_{D+1}$ uniformly at random from $N(v)$\;
			$R = \emptyset$, $E_v = \{\{v,u_i\}\}$\;
			\For{$e_1,e_2 \in E_v, e_1 \neq e_2$}{
				$E'_v = E_v \setminus \{ e_1, e_2 \}$\;
				$R = R \cup{}$ \bnddeg{$G \setminus E'_v, k, \kRec, D, b-(D-1)$}\;
				}
				\Return $R$\;
				}
		\lElseIf{$b = 0$ \emph{{and}}~$deg(v) \leq D$}{\Return $\{ (G, k-\ell) \}$}
		\lElse(\tcp*[f]{Discard~$G$}){\Return $\emptyset$}}
\caption{Generating reduced instances}
\label{alg:deg-branch}
\end{algorithm}

Algorithm \ref{alg:deg-branch} returns the set of reduced instances, the size of which is at most the number of leaves in the search tree of the branching process (inequality can result from the algorithm discarding branches in which the number of edits necessary to branch on a vertex exceeds the number of allowed deletions remaining). We now give an upper bound on the size of this set.  

\begin{lemma}\label{lemma:search-tree}
	Let~$T$ be a search tree formed by \emph{\texttt{deg-branch}}$(G, \kRec, D, k, b)$.  The number of leaves of $T$ is at most ${D+1 \choose 2}^{\kRec/(D-1)}$.
	\end{lemma}

\begin{proof}[Proof of Lemma~\ref{lemma:search-tree}]
The number of children of each interior node of $T$ is ${D+1 \choose 2}$, resulting in at most ${D+1 \choose 2}^{depth(T)}$ leaves. The depth of $T$ is limited by the second condition of the \texttt{if} on line 3 of Algorithm \ref{alg:deg-branch}. For each recursive call, $b$ is decremented by $(D - 1)$, until $b \leq D-1$. As $b$ is initally set to $\kRec$, this implies $depth(T) \leq \kRec / (D-1)$, proving the claim.
\end{proof}

Finally, we argue that at least one member of the set of reduced instances returned by \texttt{deg-branch} is equivalent to the original. Consider a solution $F$ to \textsc{$k$-Co-Path Set} in the original instance $(G,k)$. Every vertex has at most two incident edges in $G[E\backslash F]$, and since all candidates are considered at every high-degree vertex, at least one branch correctly keeps all of these edges.

\subsection{Treewidth of Reduced Instances}

Our algorithm \texttt{deg-branch} produces reduced instances with bounded degree; in order to bound their treewidth, we make use of the following result, which originated from Lemma 1 in~\cite{fomin} and was extended in~\cite{gaspers}.
\begin{lemma}
For $\epsilon > 0$, there exists $~ n_\epsilon \in \mathbb{Z}^+$ s.t. for every graph G with $n > n_\epsilon$ vertices,
\begin{equation*}
tw(G) \leq  \left(\sum_{i=3}^{17}c_in_i\right) + n_{\geq18} + \epsilon n,
\end{equation*}
where $n_i$ is the number of vertices of degree i in $G$ for i $\in \{3, \ldots, 17\}$, $n_{\geq18}$ is the number of vertices of degree at least 18, and $c_i$ is given in Table \ref{tab:coeffs}. Moreover, a tree decomposition of the corresponding width can be constructed in polynomial time in $n$. 
\label{fomin}
\end{lemma}

\
\begin{table}
\centering
\setlength{\tabcolsep}{0.5em}%
\noindent \begin{tabular}{ p{20pt} p{20pt} p{20pt} p{20pt}  p{20pt}  p{20pt} p{20pt}  p{20pt} p{20pt}  } \toprule
    d & 3 & 4 & 5 & 6 & 7 & 8 & 9 & 10 \\ 
   $c_d$ & 0.1667 & 0.3334 & 0.4334 & 0.5112 & 0.5699 & 0.6163 & 0.6538 & 0.6847  \\ \midrule
      d & 11 & 12 & 13 & 14 & 15 & 16 & 17  \\ 
       $c_d$ & 0.7105 & 0.7325 & 0.7514 & 0.7678 & 0.7822 & 0.7949 & 0.8062 \\
   \bottomrule
\end{tabular}\\
\vspace{0.5em}
\caption{Numerically obtained constants $c_d$, $3 \leq d \leq 17$, used in Lemma~\ref{fomin}; originally given in Table 6.1 of~\cite{gaspers}.}
 \label{tab:coeffs}
\end{table}

Since the structure of \textsc{$k$-Co-Path Set} naturally provides some constraints on the degree sequence of yes-instances, we are able to apply Lemma~\ref{fomin} to our reduced instances to effectively bound treewidth. We first find an upper bound on the number of degree-3 vertices in any yes-instance of \textsc{$k$-Co-Path Set}.

\begin{lemma}
Let $n_i$ be the number of vertices of degree i in a graph $G$ for any i $\in \mathbb{Z}^+$, and $\Delta$ be the maximum degree of $G$. If $(G, k)$ is a yes-instance of \textsc{$k$-Co-Path Set}, then $n_3 \leq 2k - (\sum_{i=4}^\Delta{(i-2)n_i})$.
\label{lem:n3bound}
\end{lemma}
\begin{proof}
Since $(G,k)$ is a yes-instance, removing some set of at most $k$ edges results in a graph of maximum degree $2$. For a vertex of degree $j \geq 3$, at least $j - 2$ incident edges must be removed. Thus, $n_3 + 2n_4 + 3n_5 + \ldots + (\Delta-2)n_\Delta \leq 2k$ (each removed edge counts twice -- once for each endpoint). 

\end{proof}

\begin{lemma}
Let $(G,k)$ be an instance of \textsc{$k$-Co-Path Set} such that $G$ has $n$ vertices and max degree at most $\Delta \in \{3, \ldots, 17\}$. Then the treewidth of $G$ is bounded by $k/3 + \epsilon n + c$, for some constant $c$ and any $\epsilon > 0$. A tree decomposition of the corresponding width can be constructed in polynomial time in $n$. 
\label{our_tw}
\end{lemma}

\begin{proof}
%Set $\epsilon = 0.00001$, and $N = n_\epsilon$ from Lemma \ref{fomin}. 
Let $n_\epsilon$ be defined as in Lemma \ref{fomin}.
Let $G'$ be the graph formed by adding $N = n_\epsilon$ isolates to $G$. 
%By Lemma~\ref{fomin}, $tw(G') \leq (1/6) n_3 + (1/3) n_4 + (13/30) n_5 + (23/45) n_6 + n_{\geq7} +\epsilon (N+n)$, where $n_i$ is the number of vertices of degree $i$ in $G$ (which is the same as in $G'$ since only isolates were added).
By Lemma \ref{fomin}, because $G'$ has maximum degree at most $\Delta$, $tw(G') \leq (1/6) n_3 + (1/3) n_4 + \ldots + c_{\Delta}n_{\Delta} +\epsilon (N+n)$. We can substitute the bound for $n_3$ from Lemma \ref{lem:n3bound}, which yields:
\begin{align*}
tw(G') & \leq   \frac{2k - (\sum_{i=4}^{\Delta}{(i-2)n_i})}{6} + \frac{n_4}{3} + \ldots + c_{\Delta}n_{\Delta} +\epsilon (N+n)\\
& \leq  \frac{k}{3} + \epsilon(n+N).
\end{align*}
Note that the inequality holds because we can pair the negative terms of ${(\sum_{i=4}^{\Delta}{(i-2)n_i})}/{6}$ with the corresponding terms of ${n_4}/{3} + \ldots + c_{\Delta}n_{\Delta}$ and the value of $c_jn_j - {(j-2)(n_j)}/{6}$ is non-positive for all $j \in [4,17]$.  Since $N = n_\epsilon$ is a constant, we have $tw(G') \leq k/3 + \epsilon n + c$. Since $G \subseteq G'$ and treewidth is monotone under subgraph inclusion, this proves the claim.
\end{proof}

We point out that when applying Lemma~\ref{our_tw} to reduced instances, computing the desired tree decomposition is polynomial in $k$ (since they are subgraphs of a $6k$-kernel). 

\subsection{The Algorithm \texttt{copath}}\label{sec:copathalg}
This section describes how we combine the above techniques to prove Theorem~\ref{main}. As shown in Algorithm~\ref{alg:CCCP}, we start by applying \texttt{6k-kernel}~\cite{feng2} to find $G'$, a kernel of size at most $6k$; this process deletes $k - k'$ edges. 
We then guess the number of edges $k_1 \in [0, k']$ to remove during branching, and use \texttt{deg-branch} to create a set of reduced instances $Q_{k_1}$, 
each of which have $k' - k_1$ edges. Note that \texttt{deg-branch} considers \emph{all} possible reduced instances, and thus if a (cc-)solution exists, it is contained in at least one reduced instance. To ensure the complexity of finding the reduced instances does not dominate the running time, we set the degree bound $D$ of the reduced instances to be 10 (any choice of $10 \leq D \leq 17$ is valid). By considering all possible values of $k_1$, we are assured that if $(G,k)$ is a yes-instance, some $Q_{k_1}$ contains a yes-instance. Each reduced instance is then passed to \twAlg, which correctly decides the problem with probability $2/3$. %(due to our selection of $N$ as discussed in Section 3.1). 

\begin{algorithm}[htb]
\SetKwFunction{cpsbranch}{\mainAlg}
\DontPrintSemicolon
\SetKwProg{cpsbranchprog}{Algorithm}{}{}
\cpsbranchprog{  \cpsbranch{$G$,$k$}  }{
 $(G', k') =$ \texttt{6k-kernel}$(G, k)$\; 

 \For{$k_1 \leftarrow 0$ \KwTo $k'$}{
 $Q_{k_1} =$  \texttt{deg-branch}$(G', k', k_1, 10, k_1)$\;  
 \ForEach{$(G_i, k_2) \in Q_{k_1}$}{
  \lIf{\emph{\texttt{tw-copath}}$(G_i,k_2)$}{\Return \texttt{true}}
}
} 
\Return \texttt{false}\;
}
  \caption{Deciding \textsc{$k$-Co-Path Set}}
  \label{alg:CCCP}
\end{algorithm}

\begin{proof}[Proof of Theorem \ref{main}]  
We now analyze the running time of \mainAlg, as given in Algorithm~\ref{alg:CCCP}. By Lemma~\ref{lemma:search-tree}, the size of each $Q_{k_1}$ is $O(1.561^{k_1})$. For each reduced instance $(G_i, k_2)$ in $Q_{k_1}$, we have $tw(G_i) \leq k_2/3 + \epsilon (6k) + c$ 
by Lemma~\ref{our_tw}. 

Applying Theorem~\ref{cutcount}, \twAlg runs in time $O^*(4^{k_2/3 + \epsilon 6k})$ for each reduced instance $(G_i, k_2)$ in $Q_{k_1}$ (with success probability at least 2/3). Each iteration of the outer \texttt{for} loop can then be completed in time
$$ O^*(1.561^{k_1} 4^{k_2/3 + \epsilon6k}) = O^*(4^{k/3 + \epsilon6k}) = O^*(1.588^k),$$
where we use that $k_1 + k_2 = k' \leq k$, and choose $\epsilon < 10^{-5}$. Since this loop runs at most $k+1$ times, this is also a bound on the overall computational complexity of \mainAlg. Additionally \mainAlg is linear-fpt, as the kernelization of \cite{feng2} is $O(n)$, and the kernel has size $O(k)$, avoiding any additional poly($n$) complexity from the \twAlg subroutine. Note that by Lemma \ref{fomin} the tree decomposition can be found in polynomial time in the size of the reduced instance. Since reduced instances are subsets of $6k$-kernels, the linearity is unaffected because the graph has size polynomial in $k$. 
\end{proof}

	\section{Conclusion}
	This paper gives an $O^*(4^{tw})$ fpt algorithm for \textsc{Co-Path Set}. By coupling this with kernelization and branching, we derive an $O^*(1.588^k)$ linear-fpt algorithm for deciding \textsc{$k$-Co-Path}, significantly improving the previous best-known result of $O^*(2.17^k)$. We believe that the idea of  combining a branching algorithm which guarantees equivalent instances with bounds on the degree sequence from the problem's constraints can be applied to other problems in order to obtain a bound on the treewidth (allowing treewidth-parameterized approaches to be extended to general graphs).

One natural question is whether similar techniques extend to the generalization of \textsc{Co-Path Set} to $k$-uniform hypergraphs (as treated in Zhang et al.~\cite{zhang}). It is also open whether the combined parameterization asking for a co-path set of size $k$ resulting in $\ell$ disjoint paths is solvable in sub-exponential fpt time. \looseness-1
	\section*{Acknowledgements}
	\footnotesize
This work supported in part by the Gordon \& Betty Moore Foundation under DDD Investigator Award GBMF4560 and the DARPA GRAPHS program under SPAWAR Grant N66001-14-1-4063. Any opinions, findings, and conclusions or recommendations expressed in this publication are those of the author(s) and do not necessarily reflect the views of DARPA, SSC Pacific, or the Moore Foundation. We thank two anonymous reviewers for providing a simplification of our previous branching algorithm and pointing out the result from~\cite{gaspers} enabling us to branch on vertices with degree greater than 7. We also thank Felix Reidl for helpful suggestions on an earlier draft that significantly improved the presentation of the results.
	\bibliographystyle{plain}
	{\bibliography{co-path_bib}}
	\newpage

\end{document}